\documentclass[prd, preprint, tightenlines, a4paper, eqsecnum]{revtex4-1}
\usepackage[latin1]{inputenc}
\pdfoutput=1

\usepackage{amsmath}    % need for subequations
\usepackage{amsfonts}
\usepackage{amssymb}
\usepackage{mathrsfs}   % used for script \mathscr
\usepackage{amsthm}     % need for theorem environment
\usepackage{graphicx}   % need for figures
\usepackage{pgfplots}   % need for plots
\usepackage{verbatim}   % useful for program listings
\usepackage{color}      % use if color is used in text
\usepackage{subfigure}  % use for side-by-side figures
\usepackage{hyperref}   % use for hypertext links

\theoremstyle{plain}
\newtheorem{theorem}{Theorem}[section]
\newtheorem{lemma}[theorem]{Lemma}
\newtheorem{proposition}[theorem]{Proposition}

\allowdisplaybreaks[3]

\begin{document}

%%% AUTHORS

\author{Hugo R. C. Ferreira}
\email{pmxhrf@nottingham.ac.uk}
\affiliation{School of Mathematical Sciences, University of Nottingham, Nottingham NG7 2RD, United Kingdom}

\author{Jorma Louko}
\email{jorma.louko@nottingham.ac.uk}
\affiliation{School of Mathematical Sciences, University of Nottingham, Nottingham NG7 2RD, United Kingdom}

%%% TITLE

\title{Renormalized vacuum polarization on rotating \\ warped AdS${}_3$ black holes}

%%% DATELINE

\date{Revised December 2014}

%%% ABSTRACT

\begin{abstract}
We compute the renormalized vacuum polarization of a massive scalar field in the Hartle-Hawking state on (2+1)-dimensional rotating, spacelike stretched black hole solutions to Topologically Massive Gravity, surrounded by a Dirichlet mirror that makes the state well defined. The Feynman propagator is written as a mode sum on the complex Riemannian section of the spacetime, and a Hadamard renormalization procedure is implemented by matching to a mode sum on the complex Riemannian section of a rotating Minkowski spacetime. No analytic continuation in the angular momentum parameter is invoked. Selected numerical results are given, demonstrating the numerical efficacy of the method. We anticipate that this method can be extended to wider classes of rotating black hole spacetimes, in particular to the Kerr spacetime in four dimensions.
\end{abstract}

%%% KEYWORDS

%\keywords{black hole, stability}

\maketitle

%%%%%%%%%%%%%%%%%%%%%%%%%%%%%%%%%%%%%%%%%%%%%%%%%%%%%%%%%%%%%%%%%%

%%% INTRODUCTION

\section{Introduction}
\label{sec:intro}

The study of quantum field theory on black hole spacetimes has mostly been restricted to static, spherically symmetric spacetimes. Nevertheless, there have been attempts at considering stationary black hole spacetimes, with the main focus on the Kerr spacetime \cite{Frolov:1982fr,Frolov:1984ra,Frolov:1986ut,Frolov:1989jh,Ottewill:2000qh,Duffy:2005mz}. One important task is the computation of expectation values of the renormalized stress-energy tensor for a matter field in a given quantum state \cite{Birrell1984a,Wald1994}. This has proven to be very challenging and, so far, almost all calculations have only addressed the differences between expectation values for different quantum states \cite{Duffy:2005mz} and the large field mass limit \cite{Belokogne:2014ysa}. In \cite{Steif:1993zv}, the stress-energy tensor for the rotating BTZ black hole \cite{Banados:1992wn,Banados:1992gq} was renormalized with respect to AdS${}_3$, by using the fact that the black hole corresponds to AdS${}_3$ with discrete identifications, but this method cannot be used for more general classes of rotating black hole solutions. We could summarize the main difficulties in three points: (i) the technical complexity of the computations required for the Kerr spacetime, due to the lack of spherical symmetry, (ii) the nonexistence of generalizations of the (globally defined, regular and isometry-invariant) Hartle-Hawking state  defined in static spacetimes, and (iii) the unavailability of Euclidean methods which simplify the task in static spacetimes.

To tackle point (i), we focus on a rotating black hole spacetime in 2+1 dimensions, the spacelike stretched black hole \cite{Anninos:2008fx}. This is a vacuum solution of topologically massive gravity (TMG) \cite{Deser:1982vy,Deser:1981wh}, a deformation of (2+1)-dimensional Einstein gravity, and it can be thought of as a ``warped'' version of the BTZ black hole. In contrast to the BTZ solution, the causal structure of the spacelike stretched black hole is similar to that of the Kerr spacetime~\cite{Jugeau:2010nq}. In this setting, the matter field equations can be solved in terms of hypergeometric functions, which considerably simplify the technical issues in comparison with the Kerr spacetime. These black hole solutions are known to be classically stable to massive scalar field perturbations and, in particular, classical superradiance does not give rise to superradiant instabilities~\cite{Ferreira:2013zta}. In this paper, we study a quantum scalar field on this black hole spacetime.

Concerning point (ii), the Hartle-Hawking vacuum state in the Schwarzschild spacetime is well known not to generalize to the Kerr spacetime~\cite{Kay:1988mu}. As reviewed in \cite{Ottewill:2000qh}, this is linked to the existence of a speed-of-light surface, 
outside of which no observer can corotate with the Kerr horizon. However, if we surround the Kerr hole by a mirror that is inside the speed-of-light surface, and we introduce appropriate boundary conditions at the mirror, then a Hartle-Hawing state (regular at the horizon and invariant under the isometries of the spacetime) exists inside the mirror. Further, this Hartle-Hawking state is known to be free from superradiant instabilities for a massless field \cite{Ottewill:2000qh,Duffy:2005mz,Frolov:1998wf} and the same conclusion may well extend to a massive field. In this paper we introduce a similar mirror on the $(2+1)$-dimensional spacelike stretched black hole, and we consider the similar Hartle-Hawking state inside this mirror. This $(2+1)$-dimensional Hartle-Hawking state is known to be free of superradiant instabilities for massless as well as massive fields~\cite{Ferreira:2013zta}. 

Finally, regarding point (iii), while Kerr does not admit a real section with a positive definite metric~\cite{Woodhouse:1977-complex}, it does admit a real section with a complex 
Riemannian metric to which the Feynman propagator in the Hartle-Hawking state inside a mirror can be analytically continued~\cite{Gibbons:1976ue,Brown:1990di,Moretti:1999fb}. This complex Riemannian, or ``quasi-Euclidean'', section on Kerr, hence, serves as the counterpart 
of the more familiar Euclidean (or Riemannian) section of static black hole spacetimes. 
In this paper we introduce the similar complex Riemanian section of the spacelike stretched black hole, and we exploit this section to renormalize the vacuum expectation value of a massive scalar field. The crucial point is that the complex Riemannian section of the spacelike stretched black hole has a unique Green's function, and this Green's function is expressible as a discrete mode sum whose divergence at the coincidence limit can be matched to that of a corresponding mode sum on a complex Riemannian section of a rotating flat spacetime. The renormalization procedure in the Hartle-Hawking state can, hence, be carried out using this discrete mode sum. 

In summary, in this paper we shall compute the renormalized vacuum polarization $\langle \Phi^2 (x) \rangle$ of a massive scalar field $\Phi$ in the Hartle-Hawking state on a spacelike stretched black hole surrounded by a mirror with Dirichlet boundary conditions, implementing the Hadamard renormalization prescription on the complex Riemannian section of the spacetime. In the first instance, this calculation can be taken as a warm-up for the computation of the renormalized stress-energy tensor on the spacelike stretched black hole. In the longer perspective, we believe that all the conceptual aspects of our method are applicable to wide 
classes of rotating black hole spacetimes, and in particular to Kerr in four dimensions. 
An implementation of our method in more than three dimensions will of course face new technical issues due to the more complicated functions that arise in the separation of the wave equation. 

The contents of the paper are as follows. We begin in Sec.~\ref{sec:Lorentzcase} with the quantization of a massive scalar field on the spacelike stretched black hole bounded by a mirror, including a short description of the Hadamard renormalization. In Sec.~\ref{sec:complexcase}, we outline the quasi-Euclidean method we use to obtain the complex Riemannian section of the black hole spacetime and renormalize the vacuum polarization. This is followed in Sec.~\ref{sec:NE} with the numerical evaluation of the renormalized vacuum polarization. Finally, our conclusions are presented in Sec.~\ref{sec:conclusions}. Technical steps in the analysis are deferred to five appendices. Throughout this paper we use the $(-,+,+)$ signature and units in which $\hbar = c = G = k_B = 1$.

%%%%%%%%%%%%%%%%%%%%%%%%%%%%%%%%%%%%%%%%%%%%%%%%%%%%%%%%%%%%%%%%%%

%%% SPACELIKE STRETCHED BLACK HOLES AND SCALAR FIELDS

\section{Spacelike stretched black holes and scalar fields}
\label{sec:Lorentzcase}

In this section, we first give a short description of topologically massive gravity and review the basic features of the spacelike stretched black hole solutions, including their causal structure. We then proceed to quantize the massive scalar field and outline the Hadamard renormalization procedure.

%% SPACELIKE STRETCHED BLACK HOLES

\subsection{Spacelike stretched black holes}
\label{sec:spacelikestretchedblackholes}

The (2+1)-dimensional rotating black holes we focus in this paper are vacuum solutions of topologically massive gravity, whose action is obtained by adding a gravitational Chern-Simons term to the Einstein-Hilbert action with a negative cosmological constant \cite{Deser:1982vy,Deser:1981wh}
\begin{equation}
S = S_{\text{E-H}} + S_{\text{C-S}} \, ,
\end{equation}
with
\begin{align}
S_{\text{E-H}} &= \frac{1}{16\pi G} \int d^3 x \sqrt{-g} \left( R + \frac{2}{\ell^2} \right) \, , \\
S_{\text{C-S}} &= \frac{\ell}{96\pi G \nu} \int d^3 x \sqrt{-g} \, \epsilon^{\lambda\mu\nu} \, \Gamma^{\rho}_{\lambda\sigma} \left( \partial_{\mu} \Gamma^{\sigma}_{\rho\nu} + \frac{2}{3} \Gamma^{\sigma}_{\mu\tau} \Gamma^{\tau}_{\nu\rho} \right) \, .
\end{align}
$G$ is Newton's gravitational constant, $\nu$ is a dimensionless coupling, $g$ is the determinant of the metric, $R$ is the Ricci scalar, $\ell > 0$ is the cosmological length (which will be set to $\ell \equiv 1$ from now on), $\Gamma^{\rho}_{\lambda\sigma}$ are the Christoffel symbols, and $\epsilon^{\lambda\mu\nu}$ is the Levi-Civita tensor in three dimensions.

The spacelike stretched black hole is one of the several types of warped AdS${}_3$ black hole solutions \cite{Anninos:2008fx}. Its metric, in coordinates $(t,r,\theta)$, is given by 
\begin{equation}
ds^2 = - N^2(r) dt^2 + \frac{dr^2}{4 R^2(r) N^2(r)} + R^2(r) \left( d\theta + N^{\theta}(r) dt \right)^2 \, ,
\label{eq:metricbh}
\end{equation}
with $t \in (-\infty,\infty)$, $r \in (0,\infty)$, $(t,r,\theta) \sim (t,r,\theta + 2\pi)$ and
\begin{subequations}
\begin{align}
R^2(r) &= \frac{r}{4} \left[ 3(\nu^2-1)r + (\nu^2+3)(r_+ + r_-) - 4\nu \sqrt{r_+ r_-(\nu^2+3)} \right] \, , \\
N^2(r) &= \frac{(\nu^2+3)(r-r_+)(r-r_-)}{4R^2(r)} \, , \\
N^{\theta}(r) &= \frac{2\nu r - \sqrt{r_+ r_- (\nu^2+3)}}{2 R^2(r)} \, .
\end{align}
\end{subequations}

There are outer and inner horizons at $r = r_+$ and $r = r_-$, respectively, where the coordinates $(t, r, \theta)$ become singular, and a singularity at $r=r_0$. The dimensionless coupling $\nu \in (1, \infty)$ is the warp factor, and in the limit $\nu \to 1$ the above metric reduces to the metric of the BTZ black hole in a rotating frame. More details about this black hole solution can be found in  \cite{Ferreira:2013zta,Nutku:1993eb,Gurses1994,Moussa:2003fc,Moussa:2008sj,Anninos:2008fx,Bengtsson:2005zj,Anninos:2008qb}. Here, we just describe a few relevant features.

\begin{figure}[t!]
\begin{center}
\includegraphics[scale=0.55]{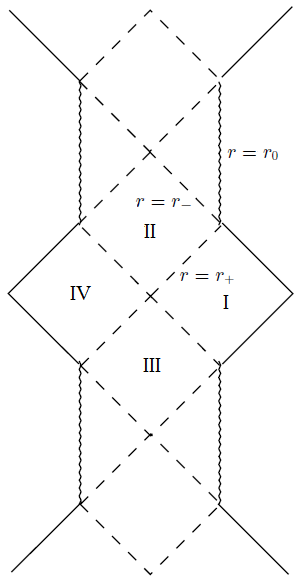} 
\hspace*{15ex}
\includegraphics[scale=0.55]{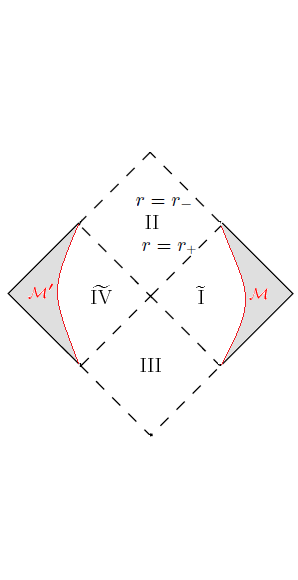}
\caption{\label{fig:CPdiagrams} Carter-Penrose diagrams of the spacelike stretched black hole spacetime for the case $r_0 < r_- < r_+$ on the left and of the manifold $M$ with mirrors described in the text on the right (adapted from Ref. \cite{Jugeau:2010nq}).}
\end{center}
\end{figure}

The Carter-Penrose diagram for this spacetime when $r_0 < r_- < r_+$ is shown in Fig.~\ref{fig:CPdiagrams}, which is essentially of the same form of those of asymptotically flat spacetimes in 3+1 dimensions. 

Consider the exterior region $r > r_+$. $\partial_t$ and $\partial_{\theta}$ are Killing vector fields. However, $\partial_t$ is spacelike \emph{everywhere}, even though surfaces of constant $t$ are still spacelike. Consequently, there is no stationary limit surface and no observers following orbits of $\partial_t$ (the usual ``static observers'' in other spacetimes) anywhere. In fact, it is easy to show that there is not any timelike Killing vector field in the exterior region of the spacetime.

Nonetheless, there are observers at a given radius $r$ following orbits of the vector field $\xi(r) = \partial_t + \Omega(r) \, \partial_{\theta}$, which are timelike as long as
\begin{equation}
\Omega_-(r) < \Omega(r) < \Omega_+(r) \, ,
\label{eq:conditionalmoststationary}
\end{equation}
with
\begin{equation}
\Omega_{\pm}(r) = - \frac{2}{2\nu r - \sqrt{r_+ r_- (\nu^2+3)} \pm \sqrt{(r-r_+)(r-r_-)(\nu^2+3)}} \, .
\end{equation}
$\Omega(r)$ is negative for all $r > r_+$, approaches zero as $r \to +\infty$, and tends to
\begin{equation}
\Omega_{\mathcal{H}} = - \frac{2}{2\nu r_+ - \sqrt{r_+ r_- (\nu^2+3)}}
\label{eq:omegaH}
\end{equation}
as $r \to r_+$. In view of these observations, we can regard $\Omega_{\mathcal{H}}$ as the angular velocity of the outer horizon with respect to stationary observers close to infinity.

One particular important class of observers are the ``locally non-rotating observers'' (LNRO), whose worldlines are everywhere normal to constant-$t$ surfaces. Because of this, they are sometimes also known as ``zero angular momentum observers'' (ZAMO). In this case, $\Omega(r) = - N^{\theta}(r)$, which satisfies \eqref{eq:conditionalmoststationary}. They are the closest to the concept of ``static observers'' in this spacetime.

Even though there is no stationary limit surface, there is still a speed-of-light surface, beyond which an observer cannot corotate with the outer horizon. Given the information above it is easy to check that the vector field $\chi = \partial_t + \Omega_{\mathcal{H}} \, \partial_{\theta}$ is the Killing vector field which generates the horizon. $\chi$ is null at the horizon and at
\begin{equation}
r = r_{\mathcal{C}} = \frac{4\nu^2 r_+ - (\nu^2+3) r_-}{3(\nu^2-1)} \, ,
\end{equation}
which is the location of the speed-of-light surface. 

In the context of quantum field theory, as it is detailed below, the nonexistence of an everywhere timelike Killing vector field in the exterior region of the spacetime is directly related to the nonexistence of a well defined quantum vacuum state which is regular at the horizon and is invariant under the isometries of the spacetime. For the Kerr spacetime, this has been proven in \cite{Kay:1988mu}. A vacuum state with these properties can however be defined if we restrict the spacetime by inserting an appropriate mirrorlike boundary which respects the Killing isometries of the spacetime. The simplest example is a boundary $\mathcal{M}$ at constant radius $r = r_{\mathcal{M}}$, in which the scalar field satisfies Dirichlet boundary conditions, $\Phi(t, r_{\mathcal{M}}, \theta) = 0$. If we choose the radius such that $r_{\mathcal{M}} \in (r_+, r_{\mathcal{C}})$, then $\chi$ is a timelike Killing vector field up to the boundary, and a vacuum state with the above properties is well defined. 
Moreover, the introduction of a mirror with reflective boundary conditions also serves to remove superradiant modes and, thus, any ambiguities they might cause when defining positive frequency mode solutions \cite{Ottewill:2000qh,Duffy:2005mz,Ferreira:2013zta}. 

For convenience, we change coordinates such that $\chi$ is given by $\chi = \partial_{\tilde{t}}$. We shall denote these ``corotating coordinates'' $(\tilde{t}=t,r,\tilde{\theta} = \theta - \Omega_{\mathcal{H}} t)$ and the metric is then given by
\begin{equation}
ds^2 = - N^2(r) d\tilde{t}^2 + \frac{dr^2}{4 R^2(r) N^2(r)} + R^2(r) \left( d\tilde{\theta} + \big( N^{\theta}(r) + \Omega_{\mathcal{H}} \big) d\tilde{t} \right)^2 \, .
\label{eq:metriccorotatingcoords}
\end{equation}

From now on, we consider as the spacetime manifold $M$ the one constructed in the following way. In region I we insert a boundary $\mathcal{M}$ at constant radius $r = r_{\mathcal{M}}$, with $r_{\mathcal{M}} \in (r_+, r_{\mathcal{C}})$, in which Dirichlet boundary conditions are imposed. We denote by $\widetilde{\text{I}}$ the portion of the region I from the horizon up to the mirror. In region IV, a similar boundary $\mathcal{M}'$ is inserted, which can be obtained by the action of a discrete isometry $J$ which takes points in region I to points in region IV by a reflection about the bifurcation surface. In a similar way, a region $\widetilde{\text{IV}}$ is defined. We take as the new manifold $M$ of interest the union of regions $\widetilde{\text{I}}$, II, III and $\widetilde{\text{IV}}$ (see Fig.~\ref{fig:CPdiagrams}).

Even though the resulting manifold is not globally hyperbolic, the Dirichlet boundary conditions imposed on the boundaries are enough to make the time evolution of the Cauchy data in any spacelike surface unique \cite{Avis:1977yn,Ishibashi:2004wx,Seggev:2003rp}. This allows us to analyze quantum field theory in this bounded spacetime.

%% FIELD EQUATION AND BASIS MODES

\subsection{Scalar field equation and basis modes}
\label{sec:fieldeq}

We consider a real massive scalar field $\Phi$ on the exterior region $\widetilde{\text{I}}$. The field obeys the Klein-Gordon equation
\begin{equation}
\left(\nabla^2 - m_0^2 - \xi R \right) \Phi = 0 \, ,
\label{eq:fieldequation1}
\end{equation}
where $m_0$ is the mass of the field, $R$ is the Ricci scalar and $\xi$ is the curvature coupling parameter. The Ricci scalar is given by $R = - 6$, which is a constant, so we can rewrite \eqref{eq:fieldequation1} as
\begin{equation}
\left(\nabla^2 - m^2 \right) \Phi = 0 \, ,
\label{eq:fieldequation2}
\end{equation}
where $m^2 \equiv m_0^2 + \xi R$ is the ``effective squared mass'' of the scalar field.

Since $\partial_{\tilde{t}}$ and $\partial_{\tilde{\theta}}$ are Killing vector fields, we consider mode solutions of \eqref{eq:fieldequation2} of the form
\begin{equation}
\Phi_{\tilde{\omega} k}(\tilde{t},r,\tilde{\theta}) = e^{-i \tilde{\omega} \tilde{t} + i k \tilde{\theta}} \, \phi_{\tilde{\omega} k}(r) \, ,
\label{eq:fieldansatz}
\end{equation}
where $\tilde{\omega} \in \mathbb{R}$ and $k \in \mathbb{Z}$.

In \cite{Ferreira:2013zta}, closed form solutions to \eqref{eq:fieldansatz} were obtained and bases of mode solutions were constructed for the unbounded spacetime. In particular, a set of ``up'' basis modes was introduced for the exterior region, corresponding to flux coming from the black hole which is partially reflected back to the black hole and partially reflected to infinity. With a boundary in place, we define a new set of modes in $\widetilde{\text{I}}$, $\Phi^{\widetilde{\text{I}}}_{\tilde{\omega}k}$, with $\tilde{\omega} > 0$, which are the unique linearly independent solutions that satisfy the Dirichlet boundary conditions at the mirror. We take these solutions to be normalized,
\begin{equation}
\langle \Phi^{\widetilde{\text{I}}}_{\tilde{\omega}k} , \Phi^{\widetilde{\text{I}}}_{\tilde{\omega}'k'} \rangle = \delta_{kk'} \, \delta(\tilde{\omega} - \tilde{\omega}') \, , 
\end{equation}
in the Klein-Gordon inner product on hypersurfaces of constant $\tilde{t}$ in $\widetilde{\text{I}}$.

With the purpose of later defining the Hartle-Hawking state, we need to construct a new mode basis. First, we define modes in the region $\widetilde{\text{IV}}$,  $\Phi^{\widetilde{\text{IV}}}_{\tilde{\omega}k}$, by the action of the discrete isometry $J$ defined previously (which takes points in region $\widetilde{\text{I}}$ to points in region $\widetilde{\text{IV}}$ by a reflection about the bifurcate surface),
\begin{equation}
\Phi^{\widetilde{\text{IV}}}_{\tilde{\omega}k} (x) := \overline{ \Phi^{\widetilde{\text{I}}}_{\tilde{\omega}k} (J^{-1} (x))} \, , \qquad x \in \widetilde{\text{IV}} \, .
\end{equation}

We then understand $\Phi^{\widetilde{\text{I}}}_{\tilde{\omega}k}$ to vanish outside region $\widetilde{\text{I}}$ and $\Phi^{\widetilde{\text{IV}}}_{\tilde{\omega}k}$ to vanish outside region $\widetilde{\text{IV}}$, and we define in the union of $\widetilde{\text{I}}$ and $\widetilde{\text{IV}}$ the new mode solutions  $\Phi^{\text{L}}_{\tilde{\omega}k}$ and $\Phi^{\text{R}}_{\tilde{\omega}k}$, by
\begin{subequations}
\begin{align}
\Phi^{\text{L}}_{\tilde{\omega}k}(x) &:= \frac{1}{\sqrt{1-e^{-2\pi \tilde{\omega} /\kappa_+}}} \left( \Phi^{\widetilde{\text{IV}}}_{\tilde{\omega}k}(x) + e^{-\pi \tilde{\omega} /\kappa_+} \overline{\Phi^{\widetilde{\text{I}}}_{\tilde{\omega}k}(x)} \right) \, , \qquad x \in \widetilde{\text{I}} \cup \widetilde{\text{IV}} \, , \\
\Phi^{\text{R}}_{\tilde{\omega}k}(x) &:= \frac{1}{\sqrt{1-e^{-2\pi \tilde{\omega} /\kappa_+}}} \left( \Phi^{\widetilde{\text{I}}}_{\tilde{\omega}k}(x) + e^{-\pi \tilde{\omega} /\kappa_+} \overline{\Phi^{\widetilde{\text{IV}}}_{\tilde{\omega}k}(x)} \right) \, , \qquad x \in \widetilde{\text{I}} \cup \widetilde{\text{IV}} \, .
\end{align}
\end{subequations}
These L and R modes can now be analytically continued to all of $M$ by crossing the horizon at $r=r_+$ in the lower half-plane in the affine parameters of the generators of the two branches at the horizon. $\Phi^{\text{L}}_{\tilde{\omega}k}$ and $\Phi^{\text{R}}_{\tilde{\omega}k}$ are, hence, of positive frequency in the affine parameters on the horizon. They are further orthonormal in the Klein-Gordon inner product on spacelike hypersurfaces from mirror to mirror (for more details of the construction, see e.g. \cite{Duffy:2005mz} 
or Appendix H of~\cite{Frolov:1998wf}).

%% QUANTIZED FIELD AND HARTLE-HAWKING STATE

\subsection{Quantized field and Hartle-Hawking vacuum state}
\label{sec:quantizedfield}

So far, only classical theory has been discussed. We now proceed to canonically quantize the scalar field using the standard Hilbert space approach. This is possible since, as seen above, there is a natural positive and negative frequency decomposition of the mode solutions for this spacetime. 

Define $\mathscr{H}$ to be the one-particle Hilbert space of the positive frequency L and R solutions, and let $\mathscr{F}_{\text{s}} (\mathscr{H})$ be the corresponding Fock space, defined in the usual way. Denote the vacuum state by $|H \rangle \in \mathscr{F}_s (\mathscr{H})$. Since the L and R solutions are positive frequency with respect to the affine parameters of the past and future horizons, this vacuum state is regular at the horizons. Furthermore, it is invariant under the spacetime isometries. Therefore, we call $|H \rangle$ the ``Hartle-Hawking vacuum state''.

The quantized scalar field $\Phi(x)$ is given by
\begin{equation}
\Phi(x) = \sum_{k=-\infty}^{\infty} \int_0^{\infty} d \tilde{\omega} \left[ a^{\text{L}}_{\tilde{\omega}k} \Phi^{\text{L}}_{\tilde{\omega}k}(x) + a^{\text{R}}_{\tilde{\omega}k} \Phi^{\text{R}}_{\tilde{\omega}k}(x) + \text{h.c.} \right] \, ,
\end{equation}
with $\Phi(x)$ being interpreted as an operator-valued distribution which acts on the Hilbert space $\mathscr{F}_s (\mathscr{H})$. The Hartle-Hawking vacuum state $|H \rangle$ satisfies
\begin{equation}
a^{\text{L}}_{\tilde{\omega}k} |H \rangle = a^{\text{R}}_{\tilde{\omega}k} |H \rangle = 0 \, .
\end{equation}

The Feynman propagator is defined as
\begin{equation}
G^{\text{F}}(x,x') := i \, \langle H | \mathscr{T} \left( \Phi(x) \Phi(x') \right) | H \rangle \, , 
\label{eq:feynmanpropagator}
\end{equation}
where $\mathscr{T}$ is the time-ordering operator. The Feynman propagator is a bidistribution, $G^{\text{F}} \in \mathcal{D}'(M \times M)$, and it is one of the Green's functions associated with the Klein-Gordon equation.

% HADAMARD RENORMALIZATION

\subsection{Hadamard renormalization}
\label{sec:hadamardrenLorentz}

The Feynman propagator, evaluated for certain quantum states, as defined in \eqref{eq:feynmanpropagator}, is a bidistribution of Hadamard type; i.e. it has a Hadamard expansion of the form
\begin{equation}
G^{\text{F}}(x,x') = \frac{i}{4 \sqrt{2} \pi} \left( \frac{U(x,x')}{\sqrt{\sigma(x,x') + i \epsilon}} + W(x,x') \right) \, , \qquad \epsilon \to 0 \! + \, .
\label{eq:GFHadamardexpansion}
\end{equation}
Here, we assume that $x$ and $x'$ belong to a geodesically convex neighborhood $N \subset M$; that is, they are linked by a unique geodesic which lies entirely in $N$. Additionally, $\sigma \in C^{\infty}(N \times N)$ is the Synge's world function, defined such that $\sigma (x,x')$ is the half of the square of the geodesic distance between $x$ and $x'$; $U \in C^{\infty}(N \times N)$ and $W \in C^{\infty}(N \times N)$ are symmetric and regular biscalar functions.

A quantum state for which the short-distance singularity structure of $G^{\text{F}}$ is given by \eqref{eq:GFHadamardexpansion} is called a ``Hadamard state''.

It can be shown (see e.g.~\cite{Decanini:2005eg}) that $U(x,x')$ only depends on the geometry along the geodesics joining $x$ to $x'$, whereas $W(x,x')$ contains the quantum state dependence of the Feynman propagator. Therefore, the singular, state-independent part of the Feynman propagator is
\begin{equation}
G_{\text{Had}}(x,x') := \frac{i}{4\sqrt{2}\pi} \frac{U(x,x')}{\sqrt{\sigma(x,x')+i\epsilon}} \, .
\end{equation}
This is known as the ``Hadamard singular part'' and it is singular at $x' \to x$.

The biscalar $U(x,x')$ can be expanded as
\begin{equation}
U(x,x') = \sum_{j=0}^{\infty} U_j(x,x') \, \sigma^j(x,x') \, .
\end{equation}
For the computation of the vacuum polarization, it is sufficient to know the zeroth term, $U_0(x,x') = 1 + \mathcal{O}(\sigma)$, thus,
\begin{equation}
G_{\text{Had}}(x,x') = \frac{i}{4\sqrt{2}\pi} \frac{1}{\sqrt{\sigma(x,x')+i\epsilon}} + \mathcal{O}(\sigma^{1/2}) \, .
\end{equation}

Given $G_{\text{Had}}$, we may obtain the renormalized vacuum polarization $\langle \Phi^2(x) \rangle$ in any Hadamard state as
\begin{equation}
\langle \Phi^2(x) \rangle := - i \lim_{x' \to x} G^{\text{F}}_{\text{ren}}(x,x') \, ,
\label{eq:phisquareddef}
\end{equation}
where
\begin{equation}
G^{\text{F}}_{\text{ren}}(x,x') := G^{\text{F}}(x,x') - G_{\text{Had}}(x,x') \, .
\end{equation}
To $\langle \Phi^2(x) \rangle$, as defined in \eqref{eq:phisquareddef}, one can add terms proportional to $m$, as can be verified by dimensional analysis. This is a usual feature of any renormalization procedure. For instance, for a scalar field of mass $m$ on Minkowski spacetime in the Minkowski vacuum $|0 \rangle$, the renormalized vacuum polarization computed as in \eqref{eq:phisquareddef} is
\begin{equation}
\langle 0 | \Phi^2(x) | 0 \rangle = - \frac{m}{4\pi} \, .
\end{equation}
We are free to set this quantity to any desired value by adding a term proportional to $m$. In the case of the expectation value of the stress-energy tensor, it is conventional to set $\langle 0 | T_{\mu\nu} (x) | 0 \rangle \equiv 0$ for the Minkowski vacuum. In this paper we shall not attempt to introduce a criterion for fixing this ambiguity and shall just define $\langle \Phi^2(x) \rangle$ as in \eqref{eq:phisquareddef}.

%%%%%%%%%%%%%%%%%%%%%%%%%%%%%%%%%%%%%%%%%%%%%%%%%%%%%%%%%%%%%%%%%%

%%% COMPLEX RIEMANNIAN SECTION

\section{Complex Riemannian section of the spacelike stretched black hole}
\label{sec:complexcase}

In this section, we first consider the complex Riemannian section of the spacelike stretched black hole and obtain the unique Green's function associated with the Klein-Gordon equation as a mode sum. This is followed by a detailed account of the Hadamard renormalization procedure, in which we subtract the divergences in the mode sum by a sum over Minkowski modes with the same singularity structure. A similar subtraction procedure in a static four-dimensional black hole spacetime in with a cosmic string has been considered in~\cite{Ottewill:2010bq}.

%% METRIC

\subsection{Complex Riemannian section}
\label{sec:complexRiemanniansection}

Euclidean methods are a powerful tool to study quantum field theory on static spacetimes. A~static spacetime is a real Lorentzian section of a complex manifold, for which it is always possible to find a real  Riemannian (or ``Euclidean'') section by performing an appropriate analytical continuation (usually a Wick rotation $t \to -i t$, where $t$ is a global timelike coordinate). In many cases, it is much easier to perform calculations in the Riemannian section (e.g. computing the unique Green's function associated with the scalar field equation) and then analytically continue the results back to the Lorentzian section.

The analytic continuation procedure does unfortunately not have an immediate generalization to spacetimes that are stationary but not static. For the exterior of a rotating black hole, one issue is that the exterior need not have a globally defined timelike vector even when each point in the exterior has a neighborhood with such a Killing vector. A~second issue is that there may exist no analytic continuation in the coordinates that results in a real section with a positive definite metric. Both of these issues are present in Kerr (for which the absence of a real section with a positive definite metric is shown in~\cite{Woodhouse:1977-complex}) and in our $(2+1)$-dimensional spacelike stretched black holes. It may be possible to obtain a positive definite metric by analytically continuing not just the coordinates but also the parameters (for continuing the angular momentum parameter in Kerr see~\cite{Hawking:1979ig}), but the physical relevance of continuing parameters seems debatable \cite{Brown:1990di} and we shall not consider such a continuation here. 

For our purposes, it will not be necessary to give an analytic continuation procedure for the full exterior region of the black hole spacetime, but just for region~$\widetilde{\text{I}}$. In this region there exists an everywhere timelike Killing vector field, $\chi = \partial_{\tilde{t}}$. If we now perform a Wick rotation $\tilde{t} = -i \tau$, with $\tau \in \mathbb{R}$, the metric \eqref{eq:metriccorotatingcoords} becomes
\begin{equation}
ds^2 = N^2(r) d\tau^2 + \frac{dr^2}{4 R^2(r) N^2(r)} + R^2(r) \left( d\tilde{\theta} - i \, \big( N^{\theta}(r) + \Omega_{\mathcal{H}} \big) d\tau \right)^2 \, .
\label{eq:metriccomplexRiemanniansection}
\end{equation}
This is the complex-valued metric $g^{\mathbb{C}}$ of the ``complex Riemannian'' (or ``quasi-Euclidean'') section $I^{\mathbb{C}}$ of a complex manifold, in which region $\widetilde{\text{I}}$ is a real Lorentzian section. This metric is regular at the horizon if $\tau$ is periodic with period $2\pi/\kappa_+$, where $\kappa_+$ is the surface gravity,
\begin{equation}
\kappa_+ = \frac{(\nu^2+3)(r_+-r_-)}{2\left(2\nu r_+ - \sqrt{(\nu^2+3)r_+ r_-}\right)}
= \frac{|\Omega_{\mathcal{H}}|}{4} (\nu^2+3)(r_+-r_-) \, .
\end{equation}
The resulting manifold has then two periodic (and thus compact) directions and a third direction that is compact by virtue of the mirror at $r = r_{\mathcal{M}}$. 

The complex Riemannian section of certain rotating spacetimes has been briefly discussed in \cite{Gibbons:1976ue} and \cite{Frolov:1982pi} in the context of the Kerr-Newman black hole. In \cite{Moretti:1999fb}, a more general concept of `local Wick rotation' is discussed for any Lorentzian manifold, even without a timelike Killing vector field, as long as its metric is a locally analytic function of the coordinates.

%% GREEN'S FUNCTION

\subsection{Green's function associated with the Klein-Gordon equation}
\label{sec:greensfunction}

In the real Lorentzian section, we defined the Feynman propagator $G^{\text{F}} \in \mathcal{D}'(\widetilde{\text{I}} \times \widetilde{\text{I}})$ as one of the Green's functions associated with the Klein-Gordon equation satisfied by the scalar field. Here, we find the Green's function $G \in \mathcal{D}'(I^{\mathbb{C}}  \times I^{\mathbb{C}} )$ associated with the Klein-Gordon equation in the complex Riemannian section. It satisfies the distributional equation
\begin{equation}
\left( \nabla^2 - m^2 \right) G(x,x') = - \frac{\delta^3(x,x')}{\sqrt{g(x)}} = - 2 \delta(\tau-\tau') \delta(r-r') \delta(\tilde{\theta}-\tilde{\theta}') \, ,
\end{equation}
where $g(x) := |\det(g^{\mathbb{C}}_{\mu\nu})|$ and $\nabla^2 := (g^{\mathbb{C}})^{\mu\nu} \nabla_{\mu} \nabla_{\nu}$ is the covariant d'Alembertian operator.

In contrast to the real Lorentzian section, in the complex Riemannian section there is a unique solution to this equation which satisfies the following boundary conditions: (i) $G(x,x')$ is regular at $r = r_+$, and (ii) $G(x,x')$ satisfies the Dirichlet boundary conditions at $r = r_{\mathcal{M}}$. This is due to the fact that two of the directions of the spacetime are periodic, while the third direction is compact. Compare this to the situation on static spacetimes without any boundary (and suitable asymptotic properties at infinity), whose Euclidean section has a unique Euclidean Green's function, due to the ellipticity of the Klein-Gordon operator.

Given the periodicity conditions of $\tau$ and $\tilde{\theta}$, one has
\begin{align}
\delta (\tau - \tau') &= \frac{\kappa_+}{2\pi} \sum_{n = -\infty}^{\infty} e^{i \kappa_+ n (\tau - \tau')} \, ,
\label{eq:deltatau} \\
\delta (\tilde{\theta} - \tilde{\theta}') &= \frac{1}{2\pi} \sum_{k = -\infty}^{\infty} e^{i k(\tilde{\theta} - \tilde{\theta}')} \, ,
\label{eq:deltatheta}
\end{align}
understood as distributional identities. If we expand $G(x,x')$ as
\begin{equation}
G(x,x') = \frac{\kappa_+}{4\pi^2} \sum_{n = -\infty}^{\infty} e^{i \kappa_+ n (\tau - \tau')} \sum_{k = -\infty}^{\infty} e^{i k(\tilde{\theta} - \tilde{\theta}')} \, G_{nk}(r,r')
\label{eq:Greenfunction0}
\end{equation}
and use \eqref{eq:deltatau} and \eqref{eq:deltatheta} one obtains a differential equation for $G_{nk}$
\begin{multline}
\frac{d}{dr} \left( 4 R^2(r) N^2(r) \frac{dG_{nk}(r)}{dr} \right) - \frac{1}{R^2(r) N^2(r)} \bigg[ R^2(r) \left( \tilde{\omega} + i k (N^{\theta}(r) + \Omega_{\mathcal{H}}) \right)^2 \\ + N^2(r) \left(k^2 + m^2 R^2(r) \right) \bigg] G_{nk}(r) = - 2 \delta(r-r') \, .
\label{eq:KGBH3}
\end{multline}
The solutions of this equation can be given in terms of solutions of the corresponding homogeneous equation. Two independent solutions of the homogeneous equation are 
\begin{align}
\phi^1_{n k} (z) &= z^{\alpha} (1-z)^{\beta} F(a,b;c;z) \, , \\
\phi^2_{n k} (z) &= z^{\alpha} (1-z)^{\beta} F(a,b;a+b+1-c;1-z) \, ,
\label{eq:exactsolution}
\end{align}
where we introduce a new radial coordinate
\begin{equation}
z = \frac{r-r_+}{r-r_-} \, ,
\label{eq:definitionz}
\end{equation}
and where the parameters of the hypergeometric functions are given by
\begin{align}
a = \alpha + \beta + \gamma \, , \qquad
b = \alpha + \beta - \gamma \, , \qquad
c = 2\alpha + 1 \, ,
\label{eq:defabc}
\end{align}
with
\begin{subequations}
\begin{align}
\alpha &= \frac{|n|}{2} \, , \\
\beta &= \frac{1}{2} + \frac{\sqrt{3(\nu^2-1)}}{\nu^2+3} \sqrt{\frac{(\nu^2+3)^2}{12(\nu^2-1)} \left( 1 + \frac{4m^2}{\nu^2+3} \right) + (\kappa_+ n + i k \Omega_{\mathcal{H}})^2} \, , \\ 
\gamma & = \frac{2\nu r_- - \sqrt{r_+ r_-(\nu^2+3)}}{(\nu^2+3)(r_+ - r_-)} \sqrt{\left[\kappa_+ n + i k \left(N^{\theta}(r_-) + \Omega_{\mathcal{H}} \right) \right]^2} \, .
\end{align}
\label{eq:alphabetagamma}
\end{subequations}
Our convention for the the branch of the square roots  in \eqref{eq:alphabetagamma} is the one with non-negative real part.

Considering again the equation \eqref{eq:KGBH3} for $G_{nk}$, the regular solution near the event horizon at $z=0$ is
\begin{equation}
p_{nk}(z) = \phi^1_{nk}(z) \, ,
\label{eq:pnk}
\end{equation}
whereas the Dirichlet solution near the mirror at $z=z_M$ is given by
\begin{equation}
q_{nk}(z) = \phi^2_{nk}(z) - \frac{\phi^2_{nk}(z_M)}{\phi^1_{nk}(z_M)} \phi^1_{nk}(z) \, .
\label{eq:qnk}
\end{equation}
The radial part of the Green's function is then
\begin{equation}
G_{nk}(z,z') = C_{nk} \, p_{nk}(z_<) \, q_{nk}(z_>) \, ,
\label{eq:Gnk}
\end{equation}
where $z_< := \min \{ z, z' \}$, $z_> := \max \{ z, z' \}$ and $C_{nk}$ is the normalization constant. If we rewrite \eqref{eq:Greenfunction0} as
\begin{equation}
G(x,x') = \frac{|\Omega_{\mathcal{H}}|}{8\pi^2} \sum_{n = -\infty}^{\infty} e^{i \kappa_+ n (\tau - \tau')} \sum_{k = -\infty}^{\infty} e^{i k(\tilde{\theta} - \tilde{\theta}')} \, G_{nk}(r,r') \, ,
\label{eq:Greenfunction1}
\end{equation}
then $C_{nk}$ is given by
\begin{equation}
C_{nk} \equiv \frac{\Gamma(a) \Gamma(b)}{|n|! \, \Gamma(a+b-|n|)} \, . 
\label{eq:Cnk}
\end{equation}
%

% HADAMARD RENORMALIZATION

\subsection{Hadamard renormalization}
\label{sec:Hadamardcomplex}

As we did before with the Feynman propagator in the real Lorentz section, we want to investigate the short-distance singularity structure of the Green's function $G$ obtained in the complex Riemannian section. That is done in some detail in Appendix \ref{appendix:geodesicstructure}, which follows \cite{Moretti:1999fb}.

The main idea is the notion of a \textit{geodesically linearly convex neighborhood} of $p \in I^{\mathbb{C}}$, which is essentially a neighborhood of $p$, $N_p \subset I^{\mathbb{C}}$, such that, for any $q, q' \in N_p$, there is only one real-parameter geodesic segment which links $q$ and $q'$ and which lies completely in $N_p$ (see Appendix \ref{appendix:geodesicstructure} for more details). It was shown in \cite{Moretti:1999fb} that, given a complex Riemannian manifold such as the one considered in this paper, for any given point, there is always a geodesically linearly convex neighborhood $N$. 

Therefore, we can define the \textit{complex Synge's world function} $\sigma \in C^{\omega} (N \times N)$, which reduces to the usual definition for Riemannian and Lorentzian manifolds. In particular, suppose we choose $x$ and $x'$ in a way such that two of their coordinates in a given coordinate system are the same and the induced metric on the submanifold defined by this condition is either Riemannian or Lorentzian. Then, we can use the previous definition as half of the square of the geodesic distance between $x$ and $x'$.

Having checked that the Synge's world function can be defined in the complex Riemannian section, we can now write the Hadamard singular part of the Green's function $G$ as
\begin{equation}
G_{\text{Had}}(x,x') = \frac{1}{4\sqrt{2}\pi} \frac{1}{\sqrt{\sigma(x,x')}} + \mathcal{O}(\sigma^{1/2}) \, .
\label{eq:Hadamardsingpart}
\end{equation}

In an analogous way to the Lorentzian case, we now subtract the Hadamard singular part from the Green's function $G$,
\begin{equation}
G_{\text{ren}}(x,x') := G(x,x') - G_{\text{Had}}(x,x') \, ,
\end{equation}
from which one obtains the vacuum polarization at $x \in \widetilde{\text{I}}$,
\begin{equation}
\langle \Phi^2(x) \rangle = \lim_{x' \to x} G_{\text{ren}}(x,x') \, .
\label{eq:vacuumpolCsection}
\end{equation}
(In a slight abuse of notation, on the rhs of the equation $x, x' \in I^{\mathbb{C}}$, such that $x \in I^{\mathbb{C}}$ is the result of a Wick rotation of $x \in \widetilde{\text{I}}$.)

By construction, the Green's function $G$ is regular at $r = r_+$, satisfies the Dirichlet boundary conditions at $r = r_{\mathcal{M}}$ and is invariant under the spacetime isometries. Therefore, after analytically continuing back to the Lorentz section, $\langle \Phi^2(x) \rangle$ as given by \eqref{eq:vacuumpolCsection} is the vacuum polarization for a scalar field in the Hartle-Hawking state.

% SUBTRACTION OF THE HADAMARD SINGULAR PART

\subsection{Subtraction of the Hadamard singular part}
\label{section:subtractionHadamard}

We have obtained the Green's function $G(x,x')$ as the mode sum \eqref{eq:Greenfunction1}, whereas $G_{\text{Had}}(x,x')$ is given in the closed form \eqref{eq:Hadamardsingpart} by the Hadamard procedure. For computational purposes, it is convenient to consider a particular choice of point separation. Consider the complex Riemannian section in $(\tau, z, \tilde{\theta})$ coordinates and suppose that $x$ and $x'$ are in the region $I^{\mathbb{C}}$ and are \emph{angularly} separated. In this case, $x$ and $x'$ are in a geodesically linearly convex neighborhood and the complex Synge's world function can be obtained for small angular separation using the standard Riemannian relation. It is given by
\begin{equation}
\sigma (x,x') = \frac{1}{2} R^2(r) (\tilde{\theta}' - \tilde{\theta})^2  + \mathcal{O}(\tilde{\theta}' - \tilde{\theta})^3 \, .
\end{equation}
Thus, the Hadamard singular part of the Green's function is
\begin{equation}
G_{\text{Had}}^{\text{BH}}(x,x') = \frac{1}{4\pi} \frac{1}{R(r)|\tilde{\theta}' - \tilde{\theta}|} + \mathcal{O}(\tilde{\theta}' - \tilde{\theta}) \, .
\end{equation}
Without loss of generality, let $x = (\tau, r, 0)$ and $x' = (\tau, r, \tilde{\theta})$, with $\tilde{\theta} > 0$, such that
\begin{equation}
G_{\text{Had}}^{\text{BH}}(x,x') = \frac{1}{4\pi} \frac{1}{R(r) \tilde{\theta}} + \mathcal{O}(\tilde{\theta}) \, .
\label{eq:HadamardsingBH}
\end{equation}

As $G(x,x')$ is known only as the mode sum \eqref{eq:Greenfunction1}, the evaluation of 
$\langle \Phi^2(x) \rangle$ as the limit \eqref{eq:vacuumpolCsection} requires $G_{\text{Had}}^{\text{BH}}$ to be rewritten as a mode sum that can be combined with \eqref{eq:Greenfunction1} so that the divergences in the coincidence limit get subtracted under the sum term by term. We shall accomplish this by comparing $G_{\text{Had}}^{\text{BH}}$ to the Hadamard singular part for a scalar field in rotating Minkowski spacetime in the complex Riemannian section, which is computed in Appendix \ref{appendix:HadamardMinkowski}. The advantage of using the Minkowski spacetime is that its symmetries allow us to compute the Green's function in both closed form and as a mode sum. 

The Hadamard singular part of the Green's function for a scalar field in the Minkowski vacuum for angularly separated points can be written as
\begin{equation}
G_{\text{Had}}^{\mathbb{M}}(x,x') = \frac{1}{4\pi} \frac{1}{\rho \tilde{\theta}} + \mathcal{O}(\tilde{\theta}) = \frac{T}{2\pi} \sum_{k=-\infty}^{\infty} \left( e^{i k \tilde{\theta}} \sum_{n=-\infty}^{\infty} \, \hat{G}_{n k}^{\mathbb{M}} (\rho,\rho') \right) - \hat{G}_{\text{reg}}^{\mathbb{M}}(x,x') + \mathcal{O}(\tilde{\theta}) \, ,
\end{equation}
where the notation is described in Appendices \ref{appendix:Minkowski} and \ref{appendix:HadamardMinkowski}. 

Suppose one identifies the leading terms of the Hadamard singular parts of both spacetimes by
\begin{equation}
\frac{1}{4\pi} \frac{1}{\rho \tilde{\theta}} \equiv \gamma(r) \, \frac{1}{4\pi} \frac{1}{R(r) \tilde{\theta}} \, ,
\end{equation}
where $\gamma(r) > 0$ is a function to be specified below. This provides a matching between the two radial coordinates:
\begin{equation}
\rho(r) = \gamma^{-1}(r) \, R(r) \, .
\label{eq:matchrhor}
\end{equation}

Given this identification, we can now write
\begin{align}
G_{\text{ren}}^{\text{BH}}(x,x') &= G^{\text{BH}}(x,x') - G_{\text{Had}}^{\text{BH}}(x,x') \notag \\
&= \sum_{k=-\infty}^{\infty} e^{i k \tilde{\theta}} \sum_{n=-\infty}^{\infty} \, \left[ \frac{|\Omega_{\mathcal{H}}|}{8\pi^2} \, G_{nk}(r,r) - \frac{T}{2\pi \gamma(r)} \, \hat{G}^{\mathbb{M}}_{nk}(\rho(r),\rho(r)) \right] \notag \\ 
&\quad + \gamma^{-1}(r) \, \hat{G}_{\text{reg}}^{\mathbb{M}}(x,x') + \mathcal{O}(\tilde{\theta}) \, .
\label{eq:Gren}
\end{align}
The parameters on the Minkowski Green's function ($\rho$, $T$, $\Omega$ and $m_{\mathbb{M}}^2$) can now be chosen such that the double sum is convergent when $\tilde{\theta} \to 0$. After this matching is performed, the vacuum polarization is just given by
\begin{align}
\langle \Phi^2 (x) \rangle &= \lim_{\tilde{\theta} \to 0} G_{\text{ren}}^{\text{BH}}(x,x') \, ,
\label{eq:vacuumpolarization}
\end{align}
which is a well-defined smooth function for $x \in \widetilde{\text{I}}$.

% FIXING OF THE MINKOWSKI FREE PARAMETERS

\subsection{Fixing of the Minkowski free parameters}
\label{sec:fixing}

At least some of the parameters of the Minkowski Green's function must be fixed such that the double sum in \eqref{eq:Gren} is convergent in the coincidence limit. To motivate the choice of the parameters, we look at the large $n$ and $k$ behavior of the summand by performing a WKB-like expansion, as explained in Appendix \ref{appendix:WKB}.

Using proposition \ref{prop:WKBterms}, one can write the asymptotic expansions
\begin{equation}
G^{\text{BH}}_{nk}(z,z) \sim \sum_{j=1}^{\infty} G^{\text{BH}(j)}_{nk}(z) \, , \qquad
\hat{G}^{\mathbb{M}}_{nk}(\rho(z),\rho(z)) \sim \sum_{j=1}^{\infty} \hat{G}^{\mathbb{M}(j)}_{nk}(\rho(z)) \, ,
\end{equation}
when $n^2 + k^2 \to \infty$, with
\begin{subequations}
\begin{align}
G^{\text{BH}(1)}_{nk}(z) &= \frac{1}{2 \chi_{\text{BH}}(z)} \, , \\
G^{\text{BH}(2)}_{nk}(z) &= -\frac{\eta^2_{\text{BH}}(z)}{4 \chi_{\text{BH}}^3(z)} - \frac{z}{16 \chi_{\text{BH}}^5(z)} \left( \frac{d \chi_{\text{BH}}^2(z)}{dz} + z \frac{d^2 \chi_{\text{BH}}^2(z)}{dz^2} \right) + \frac{5 z^2}{64 \chi_{\text{BH}}^7(z)} \left( \frac{d \chi_{\text{BH}}^2(z)}{dz} \right)^2 \, ,
\end{align}
\end{subequations}
and
\begin{subequations}
\begin{align}
\hat{G}^{\mathbb{M}(1)}_{nk}(\rho) &= \frac{1}{2 \chi_{\mathbb{M}}(\rho)} \, , \\
\hat{G}^{\mathbb{M}(2)}_{nk}(\rho) &= -\frac{\eta^2_{\mathbb{M}}(\rho)}{4 \chi_{\mathbb{M}}^3(\rho)} - \frac{\rho}{16 \chi_{\mathbb{M}}^5(\rho)} \left( \frac{d \chi_{\mathbb{M}}^2(\rho)}{d\rho} + \rho \frac{d^2 \chi_{\mathbb{M}}^2(\rho)}{d\rho^2} \right) + \frac{5 \rho^2(\rho)}{64 \chi_{\mathbb{M}}^7} \left( \frac{d \chi_{\mathbb{M}}^2(\rho)}{d\rho} \right)^2 \, ,
\end{align}
\end{subequations}
and with
\begin{subequations}
\begin{align}
\chi_{\text{BH}}^2(z) &= \frac{z}{(1-z)^2} \frac{1}{\nu^2+3} \left[ \frac{\left(\kappa_+ n + i k (N^{\theta}(z) + \Omega_{\mathcal{H}})\right)^2}{N^2(z)} + \frac{k^2}{R^2(z)} \right] \, , \\
\eta^2_{\text{BH}}(z) &= \frac{z}{(1-z)^2} \frac{m^2}{\nu^2+3} \, ,
\end{align}
\end{subequations}
and
\begin{subequations}
\begin{align}
\chi_{\mathbb{M}}^2(\rho) &= \rho^2 \left( 2 \pi T n + i k \Omega \right)^2 + k^2 \, , \\
\eta^2_{\mathbb{M}}(\rho) &= \rho^2 m_{\mathbb{M}}^2 \, .
\end{align}
\end{subequations}

This allows us to write
\begin{subequations}
\begin{align}
\frac{|\Omega_{\mathcal{H}}|}{8\pi^2} \, G^{\text{BH}}_{nk}(z,z) &= \frac{1}{8\pi^2 R(z)} \frac{1}{\sqrt{\left(n + i \frac{N^{\theta}(z) + \Omega_{\mathcal{H}}}{\kappa_+} k \right)^2 + \frac{N^2(z)}{\kappa_+ R^2(z)} k^2}} + \mathcal{O}(\chi_{\text{BH}}^{-3}) \, , \\
\frac{T}{2\pi \gamma(z)} \, \hat{G}^{\mathbb{M}}_{nk}(z,z) &= \frac{1}{8\pi^2 R(z)} \frac{1}{\sqrt{\left( n + i \frac{\Omega}{2\pi T} k \right)^2 + \frac{\gamma^2(z)}{(2\pi T)^2 R^2(z)} k^2}} + \mathcal{O}(\chi_{\mathbb{M}}^{-3}) \, .
\end{align}
\end{subequations}
The terms $\mathcal{O}(\chi^{-1})$ in the two expressions match if the parameters $T$, $\Omega$ and $\gamma(z)$ are chosen as
\begin{equation}
\gamma(z) = N(z) \, , \qquad T = \frac{\kappa_+}{2\pi} \, , \qquad \Omega = N^{\theta}(z) + \Omega_{\mathcal{H}} \, .
\label{eq:matchings}
\end{equation}
This choice corresponds to have the temperature $T$ of the scalar field in Minkowski to match the Hawking temperature of the black hole and to have the angular velocity $\Omega$ to be equal to the one measured by a locally nonrotating observer at radius $z$ in the black hole spacetime.

We now claim that, with this choice of parameters, the double sum in \eqref{eq:Gren} is convergent in the coincidence limit.

%
% THEOREM
%
\begin{theorem}
If the parameters $T$, $\Omega$ and $\gamma(z)$ are chosen as in \eqref{eq:matchings}, then
\begin{equation}
\Delta G(z) := \sum_{k = -\infty}^{\infty} \sum_{n = -\infty}^{\infty} \left[ \frac{|\Omega_{\mathcal{H}}|}{8\pi^2} \, G^{\text{BH}}_{nk}(z,z) - \frac{T}{2\pi \gamma(r)} \, \hat{G}^{\mathbb{M}}_{nk}(\rho(z),\rho(z)) \right]
\end{equation}
is finite.
\end{theorem}

%
% PROOF
%
\begin{proof}
It is enough to consider
\begin{equation}
\Delta \tilde{G}(z) := \sideset{}{'}\sum_{k,n} \left[ \frac{|\Omega_{\mathcal{H}}|}{8\pi^2} \, G^{\text{BH}}_{nk}(z,z) - \frac{T}{2\pi \gamma(r)} \, \hat{G}^{\mathbb{M}}_{nk}(\rho(z),\rho(z)) \right] \, ,
\end{equation}
where $\sum_{n,k}'$ stands for the double sum over $k$ and $n$ excluding the $k=n=0$ term.

The first terms in the WKB-like expansion cancel each other, thus
\begin{equation}
\Delta \tilde{G}(z) = \sideset{}{'}\sum_{k,n} \left[ \frac{|\Omega_{\mathcal{H}}|}{8\pi^2} \,  \left[ G^{\text{BH}(2)}_{nk}(z) + \mathcal{O}(\chi_{\text{BH}}^{-5}) \right] - \frac{T}{2\pi \gamma(z)} \, \left[ \hat{G}^{\mathbb{M}(2)}_{nk}(\rho) + \mathcal{O}(\chi_{\mathbb{M}}^{-5}) \right] \right] \, .
\end{equation}
With the choice \eqref{eq:matchings}, one has
\begin{equation}
\chi_{\mathbb{M}}^2(z) = \frac{|\Omega_{\mathcal{H}}| N(z)}{2 \kappa_+} \chi_{\text{BH}}^2(z) \, .
\end{equation}
Therefore
\begin{equation}
\Delta \tilde{G}(z) = \sideset{}{'}\sum_{k,n} \left[ \frac{\mathcal{W}(z)}{\chi_{\text{BH}}^3}  + \mathcal{O}(\chi_{\text{BH}}^{-5}) \right]  \, ,
\end{equation}
where $\mathcal{W}(z)$ does not depend on $n$ and $k$.

Note that:
\begin{align}
\sideset{}{'}\sum_{k,n} \left| \frac{\mathcal{W}(z)}{\chi_{\text{BH}}^3} \right|
&\propto \sideset{}{'}\sum_{k,n} \left| R^2(z) \left(\kappa_+ n + i k (N^{\theta}(z) + \Omega_{\mathcal{H}})\right)^2 + N^2(z) k^2 \right|^{-3/2} \notag \\
&= \sideset{}{'}\sum_{k,n} \Big\{ \left[ R^2(z) \kappa_+^2 n^2 + \left( N^2(z)- R^2(z) (N^{\theta}(z) + \Omega_{\mathcal{H}})^2 \right) k^2 \right]^2 \notag \\
&\qquad\qquad\quad\; + 4 R^4(z) (N^{\theta}(z) + \Omega_{\mathcal{H}})^2 \kappa_+^2 n^2 k^2 \Big\}^{-3/4} \notag \\
&\leq \sideset{}{'}\sum_{k,n} \left[ R^2(z) \kappa_+^2 n^2 + \left( N^2(z)- R^2(z) (N^{\theta}(z) + \Omega_{\mathcal{H}})^2 \right) k^2 \right]^{-3/2} \, .
\end{align}
In Appendix \ref{appendix:convergenceseries}, it is shown that the latter series is convergent. This proves the absolute convergence of
\begin{equation}
\sideset{}{'}\sum_{k,n} \frac{\mathcal{W}(z)}{\chi_{\text{BH}}^3} \, .
\end{equation}

Finally, since
\begin{equation}
\lim_{|\chi_{\text{BH}}| \to \infty} \frac{\left| \frac{\mathcal{W}(z)}{\chi_{\text{BH}}^3}  + \mathcal{O}(\chi_{\text{BH}}^{-5}) \right|}{\left| \frac{\mathcal{W}(z)}{\chi_{\text{BH}}^3} \right|}
= 1 \, ,
\end{equation}
the limit comparison test implies the absolute convergence of
\begin{equation}
\sideset{}{'}\sum_{k,n} \left[ \frac{\mathcal{W}(z)}{\chi_{\text{BH}}^3}  + \mathcal{O}(\chi_{\text{BH}}^{-5}) \right] \, .
\end{equation}

Therefore, we conclude that the $\Delta G(z)$ is finite.
\end{proof}

%%%%%%%%%%%%%%%%%%%%%%%%%%%%%%%%%%%%%%%%%%%%%%%%%%%%%%%%%%%%%%%%%%

%%% NUMERICAL EVALUATION OF THE VACUUM POLARIZATION

\section{Numerical evaluation of the vacuum polarization}
\label{sec:NE}

% PLOT with numerical results
%
\begin{figure}[t!]
\begin{center}
\includegraphics[scale=1]{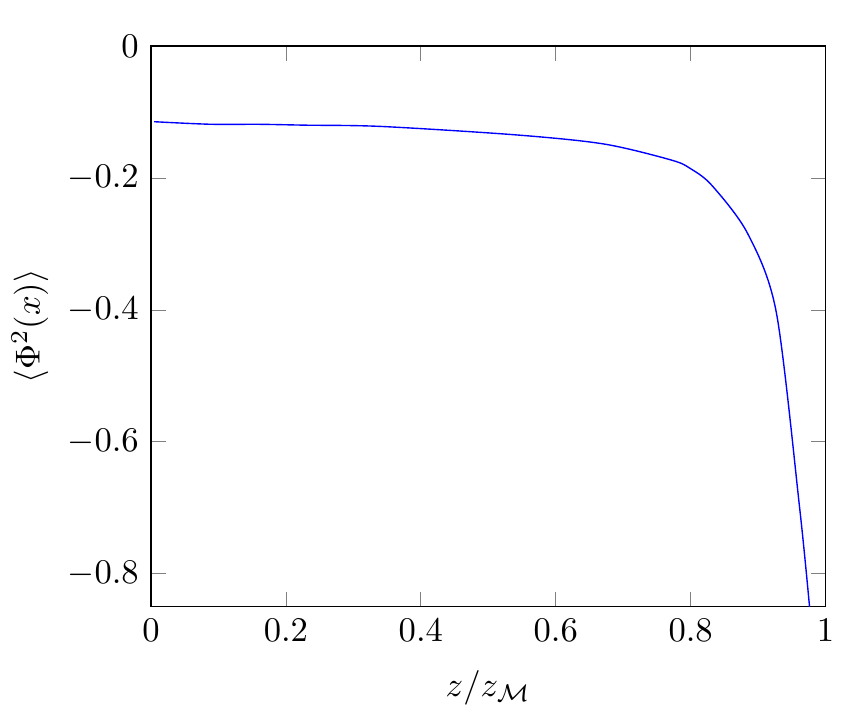} 
\caption{\label{fig:numerics1} Vacuum polarization for the scalar field as a function of $z/z_{\mathcal{M}}$ for $\nu = 1.2$, $r_+ = 15$, $r_- = 1$, $r_{\mathcal{M}} = 62$ and $m = 1$.}
\end{center}
\end{figure}

We numerically compute the vacuum polarization of the scalar field in the Hartle-Hawking state in region $\widetilde{\text{I}}$ using the expressions \eqref{eq:Gren} and \eqref{eq:vacuumpolarization} with the Minkowski parameters chosen as in \eqref{eq:matchings}:
\begin{align}
\langle \Phi^2 (x) \rangle &= \sum_{k=-\infty}^{\infty} \sum_{n=-\infty}^{\infty} \, \left[ \frac{|\Omega_{\mathcal{H}}|}{8\pi^2} \, G_{nk}(z,z) - \frac{\kappa_+}{4\pi^2 N(z)} \, \hat{G}^{\mathbb{M}}_{nk}\left(\tfrac{R(z)}{N(z)},\tfrac{R(z)}{N(z)}\right)\Big|_{\Omega = N^{\theta}(z)+\Omega_{\mathcal{H}}} \right] \notag \\ 
&\quad + \frac{1}{4\pi N(z)} \left[ {-m_{\mathbb{M}}} + \sum_{N \neq 0} \frac{e^{-m_{\mathbb{M}} \sqrt{\left(\frac{N}{T}\right)^2 - 4 \tfrac{R^2(z)}{N^2(z)} \sinh^2 \left( \frac{\Omega N}{2T} \right) + i \epsilon \, \text{sgn}(\Omega N)}}}{\sqrt{\left(\frac{N}{T}\right)^2 - 4 \tfrac{R^2(z)}{N^2(z)} \sinh^2 \left( \frac{\Omega N}{2T} \right) + i \epsilon \, \text{sgn}(\Omega N)}} \right] \, ,
\label{eq:vacuumpolcalculation}
\end{align}
with $\epsilon \to 0+$ indicating the choice of branch of the square root (see details in Appendix \ref{appendix:HadamardMinkowski}).

As described previously, the sums in \eqref{eq:vacuumpolcalculation} are convergent. For the numerical evaluation of the sums, cutoffs are imposed appropriately. Note that the parameter $m_{\mathbb{M}}^2$ is not fixed and it is chosen in such a way to improve the numerical convergence of the double sum over $k$ and $n$.

The numerical results for selected values of the parameters are presented in Fig.~\ref{fig:numerics1}. In the plot, $\langle \Phi^2(x) \rangle$ is shown as a function of the normalized radial coordinate $z/z_{\mathcal{M}}$. Note that $\langle \Phi^2(x) \rangle$ gets arbitrarily large and negative as the mirror is approached, as expected (see e.g. chapter 4.3 of \cite{Birrell1984a}). Note also that the plot is very similar to the one found in Ref.~\cite{Duffy:2002ss} for a scalar field in the (3+1)-dimensional Minkowski spacetime surrounded by a mirror with Dirichlet boundary conditions. 

We reemphasize that the result shown in Fig.~\ref{fig:numerics1} is the full renormalized vacuum polarization in the Hartle-Hawking state. To find the renormalized vacuum polarization in other Hadamard states of interest, such as the Boulware vacuum state, it would suffice to use the Hartle-Hawking state as a reference and just to calculate the difference, which is finite without further renormalization. For comparison, we note that in Kerr with a mirror the difference of the vacuum polarization in the Boulware and Hartle-Hawking states was found in \cite{Duffy:2005mz}, while the renormalized vacuum polarization in the individual states appears to be still unknown.

%%%%%%%%%%%%%%%%%%%%%%%%%%%%%%%%%%%%%%%%%%%%%%%%%%%%%%%%%%%%%%%%%%

%%% CONCLUSIONS

\section{Conclusions}
\label{sec:conclusions}

In this paper, we have computed $\langle \Phi^2(x) \rangle$ for a massive scalar field $\Phi$ in the Hartle-Hawking state on a spacelike stretched black hole with a mirror. We have employed a `quasi-Euclidean' method to obtain a complex Riemannian section of the original spacetime, in which we found the unique Green's function associated with the Klein-Gordon equation. This Green's function is given as a mode sum and its singular behavior in the coincidence limit can be subtracted by a sum over Minkowski modes with the same singularity structure. This renormalization procedure renders a smooth function whose coincidence limit is precisely the renormalized value of $\langle \Phi^2(x) \rangle$. In the future, we intend to extend this method to compute the expectation value of the stress-energy tensor.

A key ingredient in our implementation of the Hadamard renormalization was to match the mode sum for the Green's function on the complex Riemannian section of the black hole to a mode sum on the complex Riemannian section of a rotating Minkowski spacetime. We anticipate that this method can be extended to wider classes of rotating black hole spacetimes, and in particular in four dimensions to the Kerr spacetime. In Kerr, the relevant mode solutions to the Klein-Gordon equation on the complex Riemannian section would need to be constructed fully numerically, but the asymptotic properties of the solutions in the limit of large quantum numbers should be within analytic reach, and it is only these asymptotic properties that are required in the matching to mode solutions on a complex section of rotating Minkowski. Also, the freedom in the shape of the mirror in Kerr should not present complications for the matching since boundary terms in the rotating Minkowski mode functions do not enter the final subtraction terms. The implementation of our method in Kerr would, hence, seem feasible in principle, and it should prove interesting to attempt the implementation in practice.

%%%%%%%%%%%%%%%%%%%%%%%%%%%%%%%%%%%%%%%%%%%%%%%%%%%%%%%%%%%%%%%%%%

%%% ACKNOWLEDGMENTS

\begin{acknowledgements}
We thank Sam Dolan, Christopher Fewster, Bernard Kay, Ko Sanders, Peter Taylor, Helvi Witek and especially Elizabeth Winstanley for helpful discussions and comments. H.~R.~C.~F. acknowledges financial support from Funda\c{c}\~{a}o para a Ci\^{e}ncia e Tecnologia (FCT)-Portugal through Grant No.\ SFRH/BD/69178/2010. J.~L. was supported in part by STFC (Theory Consolidated Grant No.\ ST/J000388/1).
\end{acknowledgements}

%%%%%%%%%%%%%%%%%%%%%%%%%%%%%%%%%%%%%%%%%%%%%%%%%%%%%%%%%%%%%%%%%%

%%% APPENDICES

\appendix

%%% MINKOWSKI SPACETIME

\section{Rotating Minkowski spacetime in the complex Riemannian section}
\label{appendix:Minkowski}

Consider (2+1)-dimensional rotating Minkowski spacetime. Choosing rotating, spherical coordinates $(\tilde{t},\rho,\tilde{\theta})$, its metric is
\begin{equation}
ds^2 = -d\tilde{t}^2 + d\rho^2 + \rho^2 \left( d\tilde{\theta} + \Omega d\tilde{t} \right)^2 \, ,
\end{equation}
with $(\tilde{t},\rho,\tilde{\theta}) \sim (\tilde{t},\rho,\tilde{\theta}+2\pi)$. In the complex Riemannian section, the metric is given by
\begin{equation}
ds^2 = d\tau^2 + d\rho^2 + \rho^2 (d\tilde{\theta} - i \Omega d\tau)^2 \, ,
\end{equation}
with $\Omega \in \mathbb{R}$ and $t = - i \tau$.

Note that in the real Lorentzian section, for $\Omega \neq 0$, the Killing vector field $\chi = \partial_{\tilde{t}}$ becomes spacelike when $\rho > |\Omega|^{-1}$. We restrict our attention to the part of the spacetime where $\rho < \rho_{\mathcal{M}}$, such that at $\rho = \rho_{\mathcal{M}} < |\Omega|^{-1}$ there is a mirrorlike boundary at which Dirichlet boundary conditions are imposed.

Moreover, we will require that
\begin{equation}
(\tau, \rho, \tilde{\theta}) \sim (\tau + T^{-1}, \rho, \tilde{\theta}) \, ,
\end{equation}
where $T > 0$ is to be interpreted as the temperature.

Consider the Klein-Gordon equation for a real scalar field of mass $m$,
\begin{equation}
\left( \nabla^2 - m^2 \right) \Phi(\tau,\rho,\tilde{\theta}) = 0 \, ,
\label{eq:KGeq}
\end{equation}
which in this coordinate system is given by
\begin{equation}
\left[ \frac{\partial^2}{\partial\tau^2} + \frac{1}{\rho} \frac{\partial}{\partial \rho} \left( \rho \frac{\partial}{\partial \rho} \right) + \frac{1-\rho^2 \Omega^2}{\rho^2} \frac{\partial^2}{\partial\tilde{\theta}^2} + 2i \Omega \partial_{\tau} \partial_{\tilde{\theta}} - m^2 \right] \Phi (\tau,\rho,\tilde{\theta}) = 0 \, .
\end{equation}
Using the ansatz $\Phi(\tau,\rho,\theta) = e^{i \tilde{\omega} \tau + i k \tilde{\theta}} \phi(\rho)$ one gets
\begin{equation}
\frac{d^2}{d\rho^2} \phi(\rho) + \frac{1}{\rho} \frac{d}{d\rho} \phi(\rho) - \left( (\tilde{\omega}+ik\Omega)^2 + m^2 + \frac{k^2}{\rho^2} \right) \phi(\rho) = 0 \, .
\label{eq:fieldeq}
\end{equation}
Two independent solutions are
\begin{equation}
\phi^1_{\tilde{\omega} k}(\rho) = I_k \left(\sqrt{(\tilde{\omega}+ik\Omega)^2+m^2} \, \rho \right) \, , \qquad
\phi^2_{\tilde{\omega} k}(\rho) = K_k \left(\sqrt{(\tilde{\omega}+ik\Omega)^2+m^2} \, \rho \right) \, .
\end{equation}
where $I_k$ and $K_k$ are the modified Bessel functions and the principal branch of the square root is understood.

The Green's function $G(x,x')$ associated with \eqref{eq:KGeq} satisfies the equation
\begin{equation}
\left( \nabla^2 - m^2 \right) G(x,x') = - \frac{\delta^3(x,x')}{\sqrt{g(x)}} = - \frac{1}{\rho} \delta(\tau-\tau') \delta(\rho-\rho') \delta(\tilde{\theta}-\tilde{\theta}') \, .
\end{equation}

Given the periodicities of $\tau$ and $\tilde{\theta}$, $\tilde{\omega} = 2\pi T n$, with $n \in \mathbb{Z}$ and $k \in \mathbb{Z}$. One can then write
\begin{align}
\delta (\tilde{\theta} - \tilde{\theta}') &= \frac{1}{2\pi} \sum_{k=-\infty}^{\infty} e^{i k (\tilde{\theta}-\tilde{\theta}')} \, , \\
\delta(\tau-\tau') &= T \sum_{n=-\infty}^{\infty} e^{i n 2\pi T (\tau - \tau')} \, .
\end{align}

If one now expands the Green's function $G(x,x')$ as
\begin{equation}
G(x,x') = \frac{T}{2\pi} \sum_{n=-\infty}^{\infty} e^{i 2\pi T n (\tau - \tau')} \sum_{k=-\infty}^{\infty} e^{i k (\tilde{\theta}-\tilde{\theta}')} \, G_{n k} (\rho,\rho') \, ,
\label{eq:GMinkunren}
\end{equation}
then $G_{n k} (\rho,\rho')$ satisfies
\begin{equation}
\frac{d^2}{d\rho^2} G_{n k} (\rho,\rho') + \frac{1}{\rho} \frac{d}{d\rho} G_{n k} (\rho,\rho') - \left( (2\pi T n+ik\Omega)^2 + m^2 + \frac{k^2}{\rho^2} \right) G_{n k} (\rho,\rho') = - \frac{\delta(\rho-\rho')}{\rho} \, .
\label{eq:Greenfunctioneq}
\end{equation}

Consider the homogeneous equation associated with \eqref{eq:Greenfunctioneq} and let $p_{n k}(\rho)$ be the regular solution near $\rho=0$ and $q_{n k}(\rho)$ be the Dirichlet solution near $\rho=\rho_{\mathcal{M}}$. Then, the unique solution to the inhomogeneous equation is
\begin{equation}
G_{n k} (\rho,\rho') = C_{n k} \, p_{n k}(\rho_<) \, q_{n k}(\rho_>) \, ,
\end{equation}
where $C_{n k}$ is a normalization constant which is determined from the Wronskian relation
\begin{equation}
C_{n k} \left( p_{n k} \frac{d q_{n k}}{d\rho} - q_{n k} \frac{d p_{n k}}{d\rho} \right) = - \frac{1}{\rho} \, .
\label{eq:Comegak}
\end{equation}

Comparing \eqref{eq:fieldeq} and \eqref{eq:Greenfunctioneq} one concludes that the solutions to the homogeneous equation corresponding to \eqref{eq:Greenfunctioneq} are
\begin{equation}
p_{n k}(\rho) = \phi^1_{n k}(\rho) \, , \qquad 
q_{n k}(\rho) = \phi^2_{n k}(\rho) - \frac{\phi^2_{n k}(\rho_{\mathcal{M}})}{\phi^1_{n k}(\rho_{\mathcal{M}})} \phi^1_{n k}(\rho) \, .
\end{equation}

Moreover, Eq.~\eqref{eq:Comegak} leads to $C_{n k} = 1$, thus,
\begin{equation}
G_{n k} (\rho,\rho') = \phi^1_{n k}(\rho_<)  \left[ \phi^2_{n k}(\rho_>) - \frac{\phi^2_{n k}(\rho_{\mathcal{M}})}{\phi^1_{n k}(\rho_{\mathcal{M}})} \phi^1_{n k}(\rho_>) \right] \, ,
\end{equation}
with $\rho_< := \min \{ \rho, \rho' \}$ and $\rho_> := \max \{ \rho , \rho' \}$.

%%% HADAMARD SINGULAR PART FOR MINKOWSKI SPACETIME

\section{Hadamard singular part for rotating Minkowski spacetime in the complex Riemannian section}
\label{appendix:HadamardMinkowski}

In Appendix \ref{appendix:Minkowski}, the Green's function $G(x,x')$ for a scalar field at temperature $T$ in the complex Riemannian section of Minkowski spacetime was obtained as a mode sum over $k$ and $n$. Its Hadamard singular part $G_{\text{Had}}$ is given in closed form by \eqref{eq:HadamardsingBH}. For the purposes of this paper, we also want to express the Hadamard singular part of this Green's function as a mode sum.

We can write the Green's function $G(x,x')$ \eqref{eq:GMinkunren} as
\begin{equation}
G(x,x') = G_{\text{Had}}(x,x') + G_{\text{reg}}(x,x') \, ,
\end{equation}
where $G_{\text{reg}}(x,x')$ is finite when $x' \to x$. As $G_{\text{Had}}$ has no mirror dependence, it is convenient to express it as
\begin{equation}
G_{\text{Had}}(x,x') = \frac{T}{2\pi} \left( \sum_{k=-\infty}^{\infty} e^{i k (\tilde{\theta} - \tilde{\theta}')} \sum_{n=-\infty}^{\infty} e^{i n (\tau - \tau')} \, \hat{G}_{n k} (\rho,\rho') \right) -  \hat{G}_{\text{reg}}(x,x') \, ,
\label{eq:GHadMink}
\end{equation}
with
\begin{equation}
\hat{G}_{n k} (\rho,\rho') := \phi^1_{n k}(\rho) \, \phi^2_{n k}(\rho) \, ,
\end{equation}
and $\hat{G}_{\text{reg}}(x,x')$ finite when $x' \to x$. In this form, neither of the terms on the rhs of \eqref{eq:GHadMink} has any mirror dependence. We have written $G_{\text{Had}}$ as a mode sum (plus a regular term), which can be used to subtract the divergences in the black hole Green's function, as detailed in Sec.~\ref{section:subtractionHadamard}. It remains to compute $\hat{G}_{\text{reg}}(x,x')$. Since this term is finite in the coincidence limit, we only need to determine the limit of this term when $x' \to x$.

First, it will be useful to determine $G_{\text{Had}}$ in closed form. Suppose that $x$ and $x'$ are \emph{angularly} separated, i.e. $\tau = \tau'$ and $\rho = \rho'$. Then, the complex Synge's world function is given by
\begin{equation}
\sigma (x,x') = \frac{1}{2} \, \rho^2(\tilde{\theta}' - \tilde{\theta})^2 + \mathcal{O}(\tilde{\theta}' - \tilde{\theta})^3 \, .
\end{equation}
The Hadamard singular part of the Green's function is then
\begin{equation}
G_{\text{Had}}(x,x') = \frac{1}{4\pi} \frac{1}{\rho|\tilde{\theta}' - \tilde{\theta}|} + \mathcal{O}(\tilde{\theta}' - \tilde{\theta}) \, .
\end{equation}

Without loss of generality, let $x = (\tau, \rho, 0)$ and $x' = (\tau, \rho, \tilde{\theta})$, with $\tilde{\theta} > 0$, such that
\begin{equation}
G_{\text{Had}}(x,x') = \frac{1}{4\pi} \frac{1}{\rho \tilde{\theta}} + \mathcal{O}(\tilde{\theta}) \, .
\label{eq:HadamardsingMink}
\end{equation}

Note now that we can relate the thermal Green's function $G(x,x')$ at temperature $T$ to the Green's function $G_0(x,x')$ of a scalar field at zero temperature using the well-known relationship [(2.111) of \cite{Birrell1984a} or, equivalently, (2.59) of \cite{Fulling1987a}]
\begin{equation}
G(\tau, \rho, \tilde{\theta}; \, \tau', \rho', \tilde{\theta}') = \sum_{N=-\infty}^{\infty} G_0(\tau + \tfrac{N}{T}, \rho, \tilde{\theta}; \, \tau', \rho', \tilde{\theta}') \, .
\end{equation}
The zero-temperature Green's function can be written as
\begin{equation}
G_0(x,x') = \frac{1}{4\pi} \frac{e^{-m \Delta s}}{\Delta s} + G_0^{\mathcal{M}}(x,x') \, ,
\end{equation}
where $G_0^{\mathcal{M}}(x,x')$ is the contribution which contains the mirror dependence and is finite when $x' \to x$. For Minkowski spacetime in the complex Riemannian section, $\Delta s$ is given by
\begin{equation}
\Delta s^2 = (\tau'-\tau)^2 + (\rho - \rho')^2 + 4 \rho \rho' \sin^2  \left[ \frac{1}{2} \left(\tilde{\theta}' - \tilde{\theta} - i\Omega (\tau' - \tau) \right) \right] \, .
\end{equation}

In the case of angular separation, the Green's function becomes
\begin{align}
G(\tau, \rho, 0; \, \tau, \rho, \tilde{\theta}) &= \frac{1}{4\pi} \sum_{N=-\infty}^{\infty} \left[ \frac{e^{-m \sqrt{ \left(\frac{N}{T}\right)^2 + 4 \rho^2 \sin^2 \left( \frac{\tilde{\theta}}{2} + i \frac{\Omega N}{2T} \right)}}}{\sqrt{ \left(\frac{N}{T}\right)^2 + 4 \rho^2 \sin^2 \left( \frac{\tilde{\theta}}{2} + i \frac{\Omega N}{2T} \right)}} + G_0^{\mathcal{M}}(\tau + \tfrac{N}{T}, \rho, 0; \, \tau, \rho, \tilde{\theta}) \right] \notag \\
&= \frac{1}{4\pi} \left\{\rule{0ex}{6ex}\right. \sum_{N \neq 0} \left[ \frac{e^{-m \sqrt{ \left(\frac{N}{T}\right)^2 + 4 \rho^2 \sin^2 \left( \frac{\tilde{\theta}}{2} + i \frac{\Omega N}{2T} \right)}}}{\sqrt{ \left(\frac{N}{T}\right)^2 + 4 \rho^2 \sin^2 \left( \frac{\tilde{\theta}}{2} + i \frac{\Omega N}{2T} \right)}} + G_0^{\mathcal{M}}(\tau + \tfrac{N}{T}, \rho, 0; \, \tau, \rho, \tilde{\theta}) \right] \notag \\
&\qquad\qquad + \frac{e^{-2 m \rho \sin (\tilde{\theta}/2)}}{2 \rho \sin (\tilde{\theta}/2)} + G_0^{\mathcal{M}}(\tau, \rho, 0; \, \tau, \rho, \tilde{\theta})   \left.\rule{0ex}{6ex}\right\}  \notag \\
&= G_{\text{Had}}(x,x') + \hat{G}_{\text{reg}}(x,x') + G^{\mathcal{M}}(x,x')  \, ,
\end{align}
with
\begin{align}
\hat{G}_{\text{reg}}(x,x') &:= \frac{1}{4\pi} \left[ \frac{e^{-2 m \rho \sin (\tilde{\theta}/2)}}{2 \rho \sin (\tilde{\theta}/2)} - \frac{1}{\rho \tilde{\theta}} + \sum_{N \neq 0} \frac{e^{-m \sqrt{\left(\frac{N}{T}\right)^2 + 4 \rho^2 \sin^2 \left( \frac{\tilde{\theta}}{2} + i \frac{\Omega N}{2T} \right)}}}{\sqrt{ \left(\frac{N}{T}\right)^2 + 4 \rho^2 \sin^2 \left( \frac{\tilde{\theta}}{2} + i \frac{\Omega N}{2T} \right)}} \right]  \, , \label{eq:Greg} \\
G^{\mathcal{M}}(x,x') &:= \frac{1}{4\pi} \sum_{N=-\infty}^{\infty} G_0^{\mathcal{M}}(\tau + \tfrac{N}{T}, \rho, 0; \, \tau, \rho, \tilde{\theta}) \, .
\end{align}

$\hat{G}_{\text{reg}}(x,x')$ has a finite limit when $\tilde{\theta} \to 0$, except for isolated values of the parameters at which the the square root in \eqref{eq:Greg} vanishes. To see this, consider the expansion of the argument of the square root for small positive values of $\tilde{\theta}$:
\begin{align}
\left(\frac{N}{T}\right)^2 + 4 \rho^2 \sin^2 \left( \frac{\tilde{\theta}}{2} + i \frac{\Omega N}{2T} \right)
&= \left(\frac{N}{T}\right)^2 - 4 \rho^2 \left[ \sinh^2 \left( \frac{\Omega N}{2T} \right) - i \tilde{\theta} \sinh \left( \frac{\Omega N}{2T} \right) \cosh \left( \frac{\Omega N}{2T} \right) \right] \notag \\
&\quad + \mathcal{O}(\tilde{\theta})^2 \, .
\end{align}
When $\left(\frac{N}{T}\right)^2 - 4 \rho^2 \sinh^2 \left( \frac{\Omega N}{2T} \right) > 0$, the positive branch of the square root is to be used when $\tilde{\theta} \to 0$. Otherwise, when $\left(\frac{N}{T}\right)^2 - 4 \rho^2 \sinh^2 \left( \frac{\Omega N}{2T} \right) < 0$, the square root is given by
\begin{equation}
i \, \text{sgn}(\Omega N) \sqrt{4 \rho^2 \sinh^2 \left( \frac{\Omega N}{2T} \right) - \left(\frac{N}{T}\right)^2} 
= i \, \text{sgn}(\Omega) \frac{N}{T} \sqrt{\frac{4 \rho^2 T^2}{N^2} \sinh^2 \left( \frac{\Omega N}{2T} \right) - 1}  \, .
\end{equation}
Hence, one can take the limit $\tilde{\theta} \to 0$ in $\hat{G}_{\text{reg}}(x,x')$ to obtain
\begin{align}
\lim_{x' \to x} \hat{G}_{\text{reg}}(x,x') = \frac{1}{4\pi} \left[ \! -m + \sum_{N \neq 0} \frac{e^{-m \sqrt{\left(\frac{N}{T}\right)^2 - 4 \rho^2 \sinh^2 \left( \frac{\Omega N}{2T} \right) + i \epsilon \, \text{sgn}(\Omega N)}}}{\sqrt{\left(\frac{N}{T}\right)^2 - 4 \rho^2 \sinh^2 \left( \frac{\Omega N}{2T} \right) + i \epsilon \, \text{sgn}(\Omega N)}} \right] \, ,
\end{align}
with $\epsilon \to 0+$, if $\left(\frac{N}{T}\right)^2 - 4 \rho^2 \sinh^2 \left( \frac{\Omega N}{2T} \right) \neq 0$.

%%% GEODESIC STRUCTURE

\section{Geodesic structure of complex Riemannian manifolds}
\label{appendix:geodesicstructure}

To investigate the short-distance singularity structure of Green's functions in the complex Riemannian manifold, we need to verify that the local geodesic structure has been preserved when going to the complex Riemannian section. We follow \cite{Moretti:1999fb} for this analysis.

First, note that the results in \cite{Moretti:1999fb} are valid for the case in which there is a coordinate system $(x^0, x^1, \cdots)$ which covers the original Lorentzian manifold such that the metric component $g_{00} < 0$ and the inverse $g^{00} < 0$. In our case, $g_{\tilde{t} \tilde{t}} < 0$ and $g^{\tilde{t} \tilde{t}} < 0$, cf.~\eqref{eq:metriccorotatingcoords}. Additionally, the metric is required to be real analytic. 

Now, consider a complex Riemannian manifold $M^{\mathbb{C}}$ with metric $g^{\mathbb{C}}$. The geodesic equations admit locally a unique solution with parameter $t \in \mathbb{C}$ satisfying given initial conditions. If we restrict $t$ to the real domain, $t \in \mathbb{R}$, we obtain a \textit{real-parameter geodesic segment} (the corresponding complex-parameter geodesic segment is obtained by analytical continuation).

We want to define an analogous notion of geodesically convex neighborhood which is valid for the the complex Riemannian manifold. For that, we need a series of intermediate definitions. Let $\gamma_X(t)$, $t \in [0,1]$, be the real-parameter geodesic segment starting at a point $p \in M^{\mathbb{C}}$, with $X \in T_p(M^{\mathbb{C}})$ being the tangent vector to the geodesic at $p$. Let $V \subset T_p(M^{\mathbb{C}})$ be the set of vectors $X$ such that $\gamma_X(t)$ is well defined for $t \in [0,1]$. Then, we define the \textit{exponential map} $\exp_p : V \to M^{\mathbb{C}}$, $X \mapsto \gamma_X(1)$. An open \textit{star-shaped neighborhood} about $0$ of a vector space is such that, if $X$ belongs to the neighborhood, then $\lambda X$, with $\lambda \in [0,1]$, also belongs to the neighborhood. A \textit{normal neighborhood} of $p \in M^{\mathbb{C}}$ is an open neighborhood of $p$ with the form $N_p = \exp_p(S)$, with $S \subset V \subset T_p(M^{\mathbb{C}})$ an open star-shaped neighborhood of $0 \in T_p(M^{\mathbb{C}})$. A \textit{totally normal neighborhood} of $p \in M^{\mathbb{C}}$ is a neighborhood of $p$, $O_p \subset M^{\mathbb{C}}$, such that, if $q \in O_p$, there is a normal neighborhood of $q$, $N_q$, with $O_p \subset N_q$. Finally, a \textit{geodesically linearly convex neighborhood} of $p \in M^{\mathbb{C}}$ is a totally normal neighborhood of $p$, $N_p \subset M^{\mathbb{C}}$, such that, for any $q, q' \in N_p$, there is only one real-parameter geodesic segment which links $q$ and $q'$ and which lies completely in $N_p$.

It was shown in \cite{Moretti:1999fb} that, given a complex Riemannian manifold with the properties above, for any given point, there is always a geodesically linearly convex neighborhood. Therefore, we can define the \textit{complex Synge's world function} as follows. Given a geodesically linearly convex neighborhood $N \subset M^{\mathbb{C}}$, the complex Synge's world function $\sigma \in C^{\omega} (N \times N)$ is given by
\begin{equation}
\sigma (x, x') := \frac{1}{2} g(x) \left( \exp_x^{-1} (x'), \exp_x^{-1} (x') \right) \, .
\end{equation}
This reduces to the usual definition for real Riemannian and Lorentzian manifolds. In particular, suppose we choose $x$ and $x'$ in a way such that some of their coordinates in a given coordinate system are the same and the induced metric on the submanifold defined by this condition is either real Riemannian or Lorentzian. Then, we can use the usual definition as half of the square of the geodesic distance between $x$ and $x'$.

%%% WKB-LIKE EXPANSIONS

\section{WKB-like asymptotic expansion of the Green's function summand}
\label{appendix:WKB}

In this appendix, we want to study the large $n$ and $k$ behavior of the Green's function summand. For that, we obtain a WKB-like expansion of the summands, in a similar way to the approach used in \cite{Howard:1984qp}.

Let's consider the method in generality. Let $\phi^1(r)$, $\phi^2(r)$ be two independent solutions of the radial field equation, where $r$ is a radial coordinate. Suppose we define a new radial coordinate $\xi$ such that the radial field equation can be written in the form
\begin{equation}
\frac{d^2 \phi(\xi)}{d\xi^2} - \left( \chi^2(\xi) + \eta^2(\xi) \right) \phi(\xi) = 0 \, ,
\label{eq:fieldeqxi}
\end{equation}
and the Wronskian relation is given by
\begin{equation}
\phi^1(\xi) \frac{d\phi^2(\xi)}{d\xi} - \phi^2(\xi) \frac{d\phi^1(\xi)}{d\xi} = \frac{1}{C} \, ,
\label{eq:wronskianxi}
\end{equation}
where $C$ is a constant. Here, $\chi^2(\xi)$ contains all the $n$ and $k$ dependence and is large whenever $\lambda^2 := k^2+n^2$ is large. We assume then that $f(\xi; \lambda) := - \left( \chi^2(\xi) + \eta^2(\xi) \right)$ has an asymptotic expansion of the form
\begin{equation}
f(\xi; \lambda) \sim \lambda^2 \sum_{j=0}^{\infty} f_j(\xi) a_j(\lambda) \, , \qquad \lambda \to +\infty \, ,
\end{equation}
where $\{ a_j (\lambda) \}$ is an asymptotic sequence such that $a_0(\lambda) = 1$. In this case, standard WKB theory guarantees that there is an asymptotic expansion for the solutions $\phi^i(\xi)$, $i=1,2$, when $\lambda \to +\infty$, given by the so-called WKB method (see e.g. \cite{Miller:2006}).

We are interested in obtaining the large $\chi$ expansion of
\begin{equation}
\mathcal{G} (\xi) := C \, \phi^1(\xi) \, \phi^2(\xi) \, .
\end{equation}
We can obtain its asymptotic expansion in a more direct way as follows.

\begin{lemma}
$\mathcal{G}(\xi)$ satisfies the differential equation
\begin{equation}
\mathcal{G}(\xi) = \frac{1}{2 \chi(\xi)} \left[ 1 - \frac{1}{\chi^2(\xi)} \left( \mathcal{G}^{-1/2}(\xi) \frac{d^2 \mathcal{G}^{1/2} (\xi)}{d\xi^2} - \eta^2(\xi) \right) \right]^{-1/2} \, .
\end{equation}
\end{lemma}

\begin{proof}
Using \eqref{eq:fieldeqxi} and \eqref{eq:wronskianxi}, we obtain that $\mathcal{G}$ satisfies the following nonlinear equation
\begin{equation}
\frac{d^2 \mathcal{G}^{1/2} (\xi)}{d\xi^2} - \left( \chi^2(\xi) + \eta^2(\xi) \right) \mathcal{G}^{1/2} (\xi) + \frac{1}{4 \mathcal{G}^{3/2} (\xi)} = 0 \, ,
\end{equation}
which can be equivalently written as above.
\end{proof}

To obtain the large $\chi$ expansion of $\mathcal{G}(\xi)$, one can introduce an expansion parameter $\epsilon$ (which will be set to 1 at the end of the calculation)
\begin{equation}
\mathcal{G}(\xi) = \frac{1}{2 \chi(\xi)} \left[ 1 - \frac{1}{\epsilon^2 \chi^2(\xi)} \left( \mathcal{G}^{-1/2}(\xi) \frac{d^2 \mathcal{G}^{1/2} (\xi)}{d\xi^2} - \eta^2(\xi) \right) \right]^{-1/2} \, ,
\end{equation}
and expand $\mathcal{G}(\xi)$ in powers of $\epsilon^{-1}$ as
\begin{equation}
\mathcal{G}(\xi) \sim \sum_{j=1}^{\infty} \frac{\mathcal{G}^{(j)}(\xi)}{\epsilon^{2(j-1)}} \, .
\end{equation}

\begin{proposition}
The first two terms in the asymptotic expansion of $\mathcal{G}(\xi)$ are
\begin{align}
\mathcal{G}^{(1)}(\xi) = \frac{1}{2 \chi(\xi)} \, , \qquad
\mathcal{G}^{(2)}(\xi) = -\frac{\eta^2(\xi)}{4 \chi^3(\xi)} - \frac{1}{16 \chi^5(\xi)} \frac{d^2 \chi^2}{d\xi^2} + \frac{5}{64 \chi^7} \left( \frac{d \chi^2(\xi)}{d\xi} \right)^2 .
\end{align}
\label{prop:WKBterms}
\end{proposition}

\begin{proof}
Direct computation.
\end{proof}

%%% CONVERGENCE OF A SERIES

\section{Proof of convergence of a series}
\label{appendix:convergenceseries}

\begin{proposition}
Let $A$, $B>0$. Then,
\begin{align}
S := \sideset{}{'}\sum_{k,n} \frac{1}{(A n^2 + B k^2)^{3/2}} < \infty
\end{align}
where $\sum_{n,k}'$ stands for the double sum over all $k, n \in \mathbb{Z}$ excluding the $k=n=0$ term.
\end{proposition}

\begin{proof}
We write
\begin{equation}
S = \sum_{k \in \mathbb{Z}} S_k \, ,
\end{equation}
where
\begin{equation}
S_k := \begin{cases} \displaystyle\sum_{n \in \mathbb{Z}} \frac{1}{(A n^2 + B k^2)^{3/2}} \, , & k \neq 0 \, , \\
\displaystyle\sum_{n \in \mathbb{Z} \setminus \{0 \}} \frac{1}{A^{3/2} \, n^3} \, , & k = 0 \, .
\end{cases}
\end{equation}
Each $S_k$ is clearly finite, and $S_{-k} = S_k$. For $k>0$ we have
\begin{equation}
k^2 \, S_k = \sum_{n=-\infty}^{\infty} \frac{1}{\left[B + A (\frac{n}{k})^2 \right]^{3/2}} \frac{1}{k} \; \xrightarrow{\;\; k \to \infty \;\;} \; \frac{2}{B \sqrt{A}} \, ,
\label{eq:seriesRiemannsum}
\end{equation}
since the series in \eqref{eq:seriesRiemannsum} becomes the Riemann sum for the integral
\begin{equation}
\int_{-\infty}^{\infty} \frac{dt}{(B + A t^2)^{3/2}} = \frac{2}{B \sqrt{A}} \, .
\end{equation}
Thus,
\begin{equation}
S_k \sim \frac{2}{B \sqrt{A}} \frac{1}{k^2} \, , \qquad |k| \to \infty \, .
\end{equation}
so that $S$ is finite.

\end{proof}

%%% BIBLIOGRAPHY

\end{document}